\relax
\documentclass[a4paper,10pt]{article}

\usepackage{times}
\usepackage{soul}
\usepackage{url}
\usepackage[hidelinks]{hyperref}
\usepackage[utf8]{inputenc}
\usepackage[small]{caption}
\usepackage{graphicx}
\usepackage{amsmath}
\usepackage{amsthm}
\usepackage{booktabs}
\usepackage{algorithm}
\usepackage[noend]{algpseudocode}

\newtheorem{theorem}{Theorem}
\usepackage{newfloat}
\usepackage{listings}
\floatstyle{ruled}
\newfloat{listing}{tb}{lst}{}
\floatname{listing}{Listing}
\usepackage{tikz}
\usepackage{ifthen}
\usetikzlibrary{shapes.geometric, arrows}
\tikzstyle{process} = [rectangle, minimum width=3em, minimum height=2em, text centered, draw=black]
\tikzstyle{arrow} = [thick,->,>=stealth]
\usetikzlibrary{shapes,backgrounds,arrows,calc,positioning,snakes}

\usepackage{xspace}

\tikzstyle{every initial by arrow}=[initial text=] 
\tikzstyle{every state}=[fill=none,draw=black,text=black,inner sep=1pt,minimum size=1mm] 
\tikzstyle{every picture}=[->,>=stealth',shorten >=1pt,auto,node distance=1.3cm, semithick]

\tikzstyle{every node} =
[draw = none, fill = white, thin]
\tikzstyle{every edge} +=
[black, thick]

\tikzstyle{noall} =
[draw = none, fill = none]
\tikzstyle{nodraw} =
[draw = none, fill = white]
\tikzstyle{nofill} =
[draw = black, fill = none]

\tikzstyle{cnode} =
[circle, draw = black]
\tikzstyle{snode} =
[regular polygon, regular polygon sides = 4, draw = black]
\tikzstyle{lnode} =
[diamond, draw = black]

\newcommand{\Ag}{{Ag}}
\newcommand{\wf}[1]{\overline{#1}}
\newcommand{\myi}{{(i)}\xspace}
\newcommand{\myii}{{(ii)}\xspace}

\newcommand{\naww}[1]{{\langle\!\langle #1 \rangle\!\rangle}}
\newcommand{\all}[1]{{[\![ #1 ]\!]}}
\newcommand{\set}[1]{\{#1\}}

\newtheorem{example}{Example}
\newtheorem{definition}{Definition}
\newtheorem{lemma}{Lemma}

\newboolean{noednotes}
\setboolean{noednotes}{false}
\newcommand{\ednote}[1]{\ifthenelse{\boolean{noednotes}}{}{
    \marginpar{\colorbox{yellow}{\parbox{\linewidth}{\begin{center}\vspace{-3mm}\footnotesize\textsf{#1}\vspace{-3mm}\end{center}}}}}}
\newcommand{\signedednote}[2]{\ifthenelse{\boolean{noednotes}}{}{
\begin{center} \begin{minipage}{3in}
#2 --    #1
\end{minipage} \end{center}}}

\title{Towards the Verification of Strategic Properties in Multi-Agent Systems with Imperfect Information}
\author{
    Angelo Ferrando\textsuperscript{\rm 1} and
    Vadim Malvone\textsuperscript{\rm 2}\\
    {\small\textsuperscript{\rm 1}University of Genova, Italy}\\
    {\small\textsuperscript{\rm 2}T\'el\'ecom Paris, France}\\
    {\small angelo.ferrando@unige.it, vadim.malvone@telecom-paris.fr}
}

\date{}

\begin{document}

\maketitle

\begin{abstract}
In logics for the strategic reasoning the main challenge is represented by their verification in contexts of imperfect information and perfect recall. In this work, we show a technique to approximate the verification of Alternating-time Temporal Logic (ATL$\!^*$) under imperfect information and perfect recall, which is known to be undecidable. Given a model $M$ and a formula $\varphi$, we propose a verification procedure that generates sub-models of $M$ in which each sub-model $M'$ satisfies a sub-formula $\varphi'$ of $\varphi$ and the verification of $\varphi'$ in $M'$ is decidable. Then, we use CTL$\!^*$ model checking to provide a verification result of $\varphi$ on $M$. We prove that our procedure is in the same class of complexity of ATL$\!^*$ model checking under perfect information and perfect recall,
we present a tool that implements our procedure,
and provide experimental results.
\end{abstract}

\section{Introduction}

A well-known formalism for reasoning about strategic behaviours in Multi-agent Systems (MAS) is Alternating-time Temporal Logic (ATL)~\cite{AHK02}. 
An important quality of ATL is its model checking complexity, which is PTIME-complete under perfect information. However, MAS in general have imperfect information and the model checking for ATL specifications under imperfect information and perfect recall is undecidable~\cite{DimaT11}. Given the relevance of the imperfect information setting, even partial solutions to the problem can be useful. Previous approaches have either focused on an approximation to perfect information~\cite{BelardinelliLM19,BelardinelliM20} or to bounded recall~\cite{BelardinelliLM18}.
In this contribution the main idea is to modify the topological structure of the models and use CTL verification.
With more detail, given a model $M$ and a specification $\varphi$, our procedure generates a set of sub-models of $M$ in which there is perfect information and each of these sub-models satisfies a sub-formula of $\varphi$. Then, we use CTL$\!^*$ model checking to check whether: (i) the universal remaining part of $\varphi$ is satisfied, (ii) the existential remaining part of $\varphi$ is not satisfied.
In both the cases we provide a preservation result to ATL$\!^*$. 
Since the original problem is undecidable, we can not guarantee that the truth or falsity of the property can be established. In other word, we cannot provide the completeness of our solution. However, the procedure provides a constructive method to evaluate an ATL$\!^*$ formula under imperfect information and perfect recall. We also show how such a procedure is still bounded to the class of complexity of ATL$\!^*$ model checking under perfect information and perfect recall.

\paragraph{Structure of the work.} The contribution is structured as follows. 
In Section \ref{sec:preliminaries}, we present the syntax of ATL$\!^*$ as well as its semantics given on concurrent game structures with imperfect information (iCGS). In Section \ref{example:rover}, we present our running example, the curiosity rover scenario. Subsequently, in Section \ref{sec:procedure}, we show our verification procedure by presenting how to capture sub-models with perfect information that satisfy a sub-formula of the original formula and then how this procedure can be used together with CTL$\!^*$ model checking to have results. We also provide evidence of the correctness and complexity of our procedure. Furthermore, the latter is used to implement an extension of MCMAS that we discuss in Section \ref{sec:tool}. We conclude by recapping our results and pointing to future works.

\paragraph{Related work.} Several approaches for the verification of specifications in ATL under imperfect information have been recently put forward.
In one line, a translation of ATL formulas with imperfect information and imperfect recall strategies that provide a
lower and upper bounds for their truth values is presented in \cite{JamrogaKKM19}.
In another line, restrictions are made on how information is shared amongst the agents, so as to retain decidability~\cite{BerthonMMRV17}. In a related line, interactions amongst agents are limited to public actions only \cite{BelardinelliLMR17,BLMRijcai17}. These approaches are markedly different from ours as they seek to identify classes for which verification is decidable. Instead, we consider the whole class of iCGS and define a general verification procedure. In this sense, our approach is related to~\cite{BelardinelliLM18} where an incomplete bounded recall method is defined and to~\cite{BelardinelliLM19,BelardinelliM20} where an approximation to perfect information is presented.
However, while in these works perfect recall or imperfect information are approximated, here we study the topological structure of the iCGS and use CTL$\!^*$ verification to provide a verification result.

\section{Preliminaries}\label{sec:preliminaries}

In this section we recall some preliminary notions. 

Given a set $U$, $\wf{U}$ denotes its complement. 
We denote the length of a tuple $v$ as $|v|$, and its $j$-th element as $v_j$.  
For $j \leq |v|$, let $v_{\geq j}$ be the suffix $v_{j},\ldots, v_{|v|}$ of $v$ starting at $v_j$ and $v_{\leq j}$ the prefix $v_{1},\ldots, v_{j}$ of $v$.

\subsection{Model}
We start by showing a formal model for Multi-agent Systems via concurrent game structures with imperfect information \cite{AHK02,Jamroga-Hoek}.
\begin{definition}
	\label{def:cgs}
A \emph{concurrent game structure with imperfect information (iCGS)} is a tuple $M \!\!=\!\! \langle \Ag,  AP,\allowbreak S, s_I, \{Act_i\}_{i \in \Ag}, \{\sim_i\}_{i \in \Ag}, d, \delta, V \rangle$ such that: 
\begin{itemize} 
	\item $\Ag = \set{1,\dots,m}$ is a nonempty finite set of agents.
	\item $AP$ is a nonempty finite set of atomic propositions.
	\item $S\neq \emptyset$ is a finite set of {\em states}, with {\em initial state}  $s_I \in S$.
	\item For every $i \in \Ag$, $Act_i$ is a nonempty finite set of {\em actions}. Let $Act = \bigcup_{i \in \Ag} Act_i$ be the set of all actions, and $ACT = \prod_{i \in \Ag} Act_i$ the set of all joint actions. 
	\item For every $i \in \Ag$, $\sim_i$ is a relation of {\em indistinguishability} between states. That is, given states $s, s' \in S$, $s \sim_i s'$ iff $s$ and $s'$ are indistinguishable for agent $i$.
	\item The {\em protocol function} $d: \Ag \times S \rightarrow (2^{Act}\setminus \emptyset)$ defines the availability of actions so that for every $i \in Ag$, $s \in S$, \myi $d(i,s) \subseteq Act_i$ and \myii $s \sim_i s'$ implies $d(i,s) = d(i,s')$.
	\item The {\em transition function} $\delta : S \times ACT \to S$ assigns a successor state $s' = \delta(s, \vec{a})$ to each state $s\in S$, for every joint action $\vec{a} \in ACT$ such that $a_i \in d(i,s)$ for every $i \in \Ag$.
	\item $V: S \rightarrow 2^{AP}$ is the {\em labelling function}.
 \end{itemize}
\end{definition}

By Def.~\ref{def:cgs} an iCGS describes the interactions of a group $\Ag$ of agents, starting from the initial state $s_I \in S$, according to the transition function $\delta$. The latter is constrained by the availability of actions to agents, as specified by the protocol function $d$. Furthermore, we assume that agents can have imperfect information of the exact state of the game; so in any state $s$, agent $i$ considers epistemically possible all states $s'$ that are $i$-indistinguishable from $s$ \cite{FHMV95}.
When every $\sim_i$ is the identity relation, \textit{i.e.}, $s \sim_i s'$ iff $s = s'$, we obtain a standard CGS with perfect information \cite{AHK02}. 


A history $h \in S^+$ is a finite (non-empty) sequence of states.  
The indistinguishability relations are extended to histories in a synchronous, point-wise way, \textit{i.e.}, histories $h, h' \in S^+$ are {\em
indistinguishable} for agent $i \in \Ag$, or $h \sim_i h'$, iff \myi $|h| = |h'|$ and \myii for all $j \leq |h|$, $h_j \sim_i h'_j$.

\subsection{Syntax}
We use ATL$\!^*$~\cite{AHK02} to reason about the strategic abilities of agents in iCGS.
\begin{definition}
	\label{def:ATL*}
State ($\varphi$) and path ($\psi$) formulas in ATL$\!^*$ are defined as
follows:
\begin{eqnarray*}
\varphi & ::= & q \mid \neg \varphi  \mid \varphi \land \varphi \mid \naww{\Gamma} \psi\\
\psi & ::= & \varphi \mid \neg \psi \mid \psi \land \psi \mid X \psi \mid (\psi U \psi)
\end{eqnarray*}
where $q \in AP$ and $\Gamma \subseteq \Ag$.

Formulas in ATL$\!^*$ are all and only the state formulas.
\end{definition}

As usual, a formula $\naww{\Gamma} \Phi$ is read as ``the agents in coalition $\Gamma$ have a strategy to achieve $\Phi$''. The meaning of temporal operators {\em next} $X$ and {\em until} $U$ is standard \cite{BaierKatoen08}.  Operators $\all{\Gamma}$, {\em release} $R$, {\em eventually} $F$, and {\em globally} $G$ can be introduced as usual.

Formulas in the $ATL$ fragment of ATL$\!^*$ are obtained from
Def.~\ref{def:ATL*} by restricting path formulas as follows:
\begin{eqnarray*}
\psi & ::= & X \varphi \mid (\varphi U \varphi) \mid (\varphi R \varphi)
\end{eqnarray*}
In the rest of the paper, we will also consider the syntax of ATL$^*$ in negative normal form (NNF):

	\begin{eqnarray*}
		\varphi & ::= & q \mid \neg q  \mid \varphi \land \varphi \mid \varphi \vee \varphi \mid \naww{\Gamma} \psi \mid \all{\Gamma} \psi\\
		\psi & ::= & \varphi \mid \psi \land \psi \mid \psi \vee \psi \mid X \psi \mid (\psi U \psi) \mid (\psi R \psi)
	\end{eqnarray*}

	where $q \in AP$ and $\Gamma \subseteq \Ag$.


\subsection{Semantics}
We assume that agents employ {\em uniform strategies} \cite{Jamroga-Hoek}, \textit{i.e.}, they perform the same action whenever they have the same information. 

\begin{definition}
	\label{uniformitybouned}
	A \emph{uniform strategy} for agent $i \in \Ag$ is a function $\sigma_i\!: \!S^+ \!\!\to\!\! Act_i$ such that for all histories $h, h' \!\in\! S^+ \!$, \myi
	$\sigma_i(h) \!\in\! d(i, last(h))$ and \myii
	$h \!\sim_i\! h'$ implies $\sigma_i(h) \!=\! \sigma_i(h')$.  
\end{definition}

By Def.~\ref{uniformitybouned} any strategy for agent $i$ has to return actions that are enabled for $i$. Also, whenever two histories are indistinguishable for $i$, then the same action is returned. Notice that, for the case of perfect information, condition \myii is satisfied by any strategy $\sigma$. Furthermore, we obtain memoryless (or imperfect recall) strategies by considering the domain of $\sigma_i$ in $S$, \textit{i.e.} $\sigma_i: S \to Act_i$.

Given an iCGS $M$, a {\em path} $p \in S^{\omega}$ is an infinite sequence of states.  
Given a joint strategy $\Sigma_\Gamma$, comprising of one strategy for each agent in coalition $\Gamma$, a path $p$ is \emph{$\Sigma_\Gamma$-compatible} iff for every $j \geq 1$, $p_{j+1} = \delta(p_j, \vec{a})$ for some joint action $\vec{a}$ such that for every $i \in \Gamma$, $a_i = \sigma_i(p_{\leq j})$, and for every $i \in \wf{\Gamma}$, $a_i \in d(i,p_j)$.
We denote with $out(s, \Sigma_\Gamma)$ the set of all $\Sigma_\Gamma$-compatible paths from $s$.

Now, we have all the ingredients to assign a meaning to ATL$\!^*$ formulas on iCGS. 
\begin{definition}
	\label{satisfaction}
The satisfaction relation $\models$ for an iCGS $M$, state $s \in S$, path $p \in S^{\omega}$, atom $q \in AP$, and ATL$\!^*$ formula $\phi$ is defined as follows (clauses for Boolean connectives are immediate and thus omitted):

\begin{tabbing}
$(M, s) \models q$ \ \ \ \ \ \ \ \ \ \ \=  iff \ \ \ \ \= $q \in {V}(s)$\\
$(M, s) \models \naww{{\Gamma}} \psi$  \>iff\> for some joint strategy $\Sigma_\Gamma$, \\ \> \> for all $p \!\in\! out(s, \Sigma_\Gamma)$,  $(M, p) \!\models\! \psi$\\
$(M, p) \models \varphi $ \>iff\> $(M, p_1) \models \varphi$ \\
$(M, p) \models X \psi$ \>iff\> $(M, p_{\geq 2}) \models \psi$\\
$(M, p) \models \psi U \psi'$  \>iff\> for some $k \geq 1$, $(M, p_{\geq k}) \!\models\! \psi'$, and \\ \> \> for all $1 \!\leq j <\! k, (M, p_{\geq j}) \!\models\! \psi$\\
\end{tabbing}
\end{definition}

We say that formula $\varphi$ is {\em true} in an iCGS $M$, or $M \models \varphi$, iff $(M, s_I) \models \varphi$.

To conclude, we state the model checking problem.
\begin{definition}
	Given an iCGS $M$ and a formula $\varphi$, the model checking problem concerns determining whether $M \models \varphi$.
\end{definition}

Since the semantics provided in Def.~\ref{satisfaction} is the standard interpretation of ATL$\!^*$ \cite{AHK02,Jamroga-Hoek}, it
is well known that model checking ATL, {\em a fortiori} ATL$\!^*$, against iCGS with imperfect information and  perfect recall is undecidable~\cite{DimaT11}.  


\section{Curiosity rover scenario}\label{example:rover}

In this section we present a variant of the Curiosity rover scenario under imperfect information and perfect recall. 
The Curiosity rover is one of the most complex systems successfully deployed in a planetary exploration mission to date. Its  main objectives include recording image data and collecting soil/rock data. Differently from the original \cite{Curiosity}, in this example the rover is equipped with decision making capabilities, which make it autonomous.
We simulate an inspection mission, where the Curiosity patrols a topological map of the surface of Mars.

In Figure~\ref{fig:curiosity}, we model an example of mission for the rover. The rover starts in state $s_I$, and it has to perform a setup action, consisting in checking the three main rover's components: arm, mast, or the wheels. To save time and energy, the mechanic (the entity capable of performing the setup checks and making any corrections) can perform only one setup operation per mission. Since such operation is not known by the rover, from its point of view, states $s_1$, $s_2$, $s_3$ are equivalent, i.e. it cannot distinguish between the three setup operations.
After the selection, the mechanic can choose to check and eventually correct the component (action $ok$) or to decline the operation (action $nok$). In the former case the rover can continue the mission while in the latter the mission terminates with an error.
In case the mechanic collaborates, the rover can start the mission. We consider with state $s_4$ the fact that the rover is behind the ship. So, its aim is to move from its position, to make a picture of a sample rock of Mars and then to return to the initial position. More in detail, from the state $s_4$, it can decide to move left ($L$) moving to state $s_6$, or right ($R$) moving to state $s_7$. In these two states the rover can make a photo of a sample rock (action $mp$). After this step, the rover has to conclude the mission by returning behind the ship. To accomplish the latter, it needs to do the complement moving action to go back to the previously visited state ($R$ from $s_6$, and $L$ from $s_7$).

Formally, the model of the mission is defined as: $M = \langle \Ag, AP, S, s_I, \{Act_i\}_{i \in \Ag}, \allowbreak \{\sim_i\}_{i \in \Ag},  d,  \delta,  V \rangle, \textit{ where } \Ag = \{rover, mechanic\}, S = \{ s_I,  s_1, s_2, s_3, s_4, s_5, s_6, \allowbreak s_7, s_8, e_1, e_2 \}, Act_{rover} = \{ chk, L, \allowbreak R, mp,  i \}, Act_{mechanic} = \{ ca, cm, cw,  ok, nok, \allowbreak i \}, \textit{ and } \allowbreak AP =\{ sp, cp, oc, rm, rp, ip, pl, pr\}$.
The atomic propositions have the following meaning: $sp$ stays for starting position, $cp$ stays for check phase, $oc$ stays for ok check phase, $rm$ stays for ready to mission, $rp$ stays for ready to make a picture, $ip$ stays for in position, $pl$ stays for picture left, and $pr$ stays for picture right. The labelling function is defined as follows: $V(s_I) = \{ sp \}$, $V(s_1) = V(s_2) = V(s_3) = \{ cp \}$, $V(s_4) = \{ oc,rm \}$, $V(s_5) = \{ pl\}$, $V(s_6) = V(s_7) = \{ rp, ip \}$, $V(s_8) = \{ pr \}$, and $V(e_1) = V(e_2) = \emptyset$.
There are the following non trivial equivalence relations $s_1 \sim_{rover} s_2$, $s_1 \sim_{rover}s_3$, and $s_2 \sim_{rover} s_3$. 
Finally, the functions $d$ and $\delta$ can be derived by Figure~\ref{fig:curiosity}. 

So, given the model $M$ we can define several specifications. For example, the specification that describes the rover mission is $\varphi_1 = \naww{rover} F ((oc \land rm) \land  F ((pl \vee pr) \land F (oc \land rm)))$. In words the latter formula means that there exists a strategy for the rover that sooner or later it can be ready to start the mission, can make a picture of a sample rock, and can return behind the ship.
In the latter formula, we assume that the mechanic can decide not to cooperate. In this case, it is impossible for the rover to achieve the end of the mission even using memoryfull strategies.
If we assume the cooperation of the mechanic (i.e. the mechanic checks a component and eventually corrects it via action $ok$), we can rewrite the specification as $\varphi_2 = \naww{rover, mechanic} F ((oc \land rm) \land  \naww{rover} F ((pl \vee pr) \land F (oc \land rm)))$.
In this case, by using memoryless strategies the formula $\varphi_2$ still remains false since the rover can select only one action on states $s_6$ and $s_7$, so it is impossible for it to first make the picture and then go to the end of the mission. 
While, by using memoryfull strategies, the rover has a strategy to make the formula true by considering the cooperation of the mechanic. In fact, a simple strategy $\sigma$ can be generated by: $\sigma(s_I)= chk$, $\sigma(s_I s_j)=i$, $\sigma(s_I s_j s_4) = L$, $\sigma(s_I s_j s_4 s_6) = mp$, $\sigma(s_I s_j s_4 s_6 s_5) = i$, and $\sigma(s_I s_j s_4 s_6 s_5 s_6) = R$, where $j \in \{1,2,3\}$.
With the above observations, we can conclude that, to verify the specification $\varphi_2$ in the model $M$, we need of the model checking of ATL in the context of imperfect information and perfect recall, a problem in general undecidable. 
Other interesting specifications can involve specific rover's status. For example, we could be interested to verify if there exists a strategy for the rover in coalition with the mechanic such that it can be ready to make a picture but it is not in position to do that and, from that point, if it has the ability to make a picture and return in the starting mission, i.e. we want to verify the formula $\varphi_3 = \naww{rover, mechanic} F (rp \land \neg ip \land \naww{rover} F ( (pl \vee pr) \land F (oc \land rm)))$. 
It is easy to see that $\varphi_3$ is false in the model $M$ because there is not a state in which $rp \land \neg ip$. However, to check $\varphi_3$ in the model $M$, we need of the model checking of ATL in the context of imperfect information and perfect recall, a problem in general undecidable.

In the next section we will present a sound but not complete procedure to handle this class of problems and use the curiosity rover scenario to help the reader during each step.

\begin{figure}[t]\centering
	\scalebox{0.80}{
		\begin{tikzpicture}
		\node [cnode]
		(S0)
		{\footnotesize $\stackrel{sp}{s_I}$};
		\node [cnode]
		(S3)
		[below of = S0, node distance = 8em]
		{\footnotesize $\stackrel{cp}{s_3}$};
		\node [cnode]
		(S1)
		[below left of = S0, node distance = 6em]
		{\footnotesize $\stackrel{cp}{s_1}$};
		\node [cnode]
		(S2)
		[below right of = S0, node distance = 6em]
		{\footnotesize $\stackrel{cp}{s_2}$};
		\node [cnode]
		(S4)
		[below of = S3, node distance = 5em]
		{\footnotesize $\stackrel{oc,rm}{s_4}$};
		\node [cnode]
		(E2)
		[below of = S4, node distance = 7em]
		{\footnotesize $\stackrel{\emptyset}{e_2}$};
		\node [cnode]
		(S6)
		[left of = E2, node distance = 6em]
		{\footnotesize $\stackrel{rp,ip}{s_6}$};
		\node [cnode]
		(S7)
		[right of = E2, node distance = 6em]
		{\footnotesize $\stackrel{rp,ip}{s_7}$};
		\node [cnode]
		(S5)
		[left of = S6, node distance = 7em]
		{\footnotesize $\stackrel{pl}{s_5}$};
		\node [cnode]
		(S8)
		[right of = S7, node distance = 7em]
		{\footnotesize $\stackrel{pr}{s_8}$};
		\node [cnode]
		(E1)
		[above of = S8, node distance = 12em]
		{\footnotesize $\stackrel{\emptyset}{e_1}$};

		\path[-stealth']
		(S4)  edge  [pos = 0.5, bend right = 30]
		node [above = -0.2] {\tiny $(L,i)$}
		(S6)
		edge  [pos = 0.5, bend left = 30]
		node [above = -0.2] {\tiny $(R,i)$}
		(S7)
		edge	[pos = 0.5, loop below]
		node [] {\tiny $(i,i)$}
		()
		
		(S0)
		edge    [pos = 0.5, bend right =30]
		node    [above = -0.2] {\tiny $(chk,ca)$}
		(S1)
		edge    [pos = 0.25]
		node    [above = -0.2] {\tiny $(chk,cm)$}
		(S3)
		edge    [pos = 0.5, bend left = 30]
		node    [above = -0.2] {\tiny $(chk,cw)$}
		(S2)
		edge	[pos = 0.5, loop above]
		node [] {\tiny $(i,\star)$}
		()
		
		(S1)
		edge    [pos = 0.5, bend right = 15]
		node    [above = -0.2] {\tiny $(i,ok)$}
		(S4)
		edge    [pos = 0.7]
		node    [above = -0.2] {\tiny $(i,nok)$}
		(E1)
		
		(S3)
		edge    [pos = 0.4]
		node    [above = -0.2] {\tiny $(i,ok)$}
		(S4)
		edge    [pos = 0.5]
		node    [above = -0.2] {\tiny $(i,nok)$}
		(E1)
		
		(S2)
		edge    [pos = 0.5, bend left = 15]
		node    [above = -0.2] {\tiny $(i,ok)$}
		(S4)
		edge    [pos = 0.4]
		node    [above = -0.2] {\tiny $(i,nok)$}
		(E1)
		
		(S6)	
		edge	[pos = 0.4]
		node [above = -0.2] {\tiny $(L,i)$}
		(E2)
		edge	[pos = 0.4, bend right = 15]
		node [above = -0.2] {\tiny $(mp,i)$}
		(S5)
		edge	[pos = 0.4, bend right = 15]
		node [above = -0.2] {\tiny $(R,i)$}
		(S4)
		
		(E2)	
		edge	[pos = 0.5, loop above]
		node [] {\tiny $(\star,\star)$}
		()
		
		(E1)	
		edge	[pos = 0.5, loop below]
		node [] {\tiny $(\star,\star)$}
		()
		
		(S7)	edge	[pos = 0.4]
        node [above = -0.2] {\tiny $(R,i)$}
		(E2)
		edge	[pos = 0.4, bend left = 15]
		node [above = -0.2] {\tiny $(mp,i)$}
		(S8)
		edge	[pos = 0.4, bend left = 15]
		node [above = -0.2] {\tiny $(L,i)$}
		(S4)
		
		(S5)
		edge	[pos = 0.5, bend right = 15]
		node [above = -0.2] {\tiny $(i,i)$}
		(S6)

		(S8)
		edge	[pos = 0.5, bend left = 15]
		node [above = -0.2] {\tiny $(i,i)$}
		(S7)

		;
		
		\draw[dotted, -] (S1) -- (S3);
		
		\draw[dotted, -] (S3) -- (S2);
		
		\draw[dotted, -] (S1) -- (S2);
		
%
%

		\end{tikzpicture}
	} 
	 \caption{An example of the rover's mission where $i$ stands for the idle action, $\star$ for any action, and the rover's equivalence relation is denoted with the dotted lines.} \label{fig:curiosity} 
\end{figure}
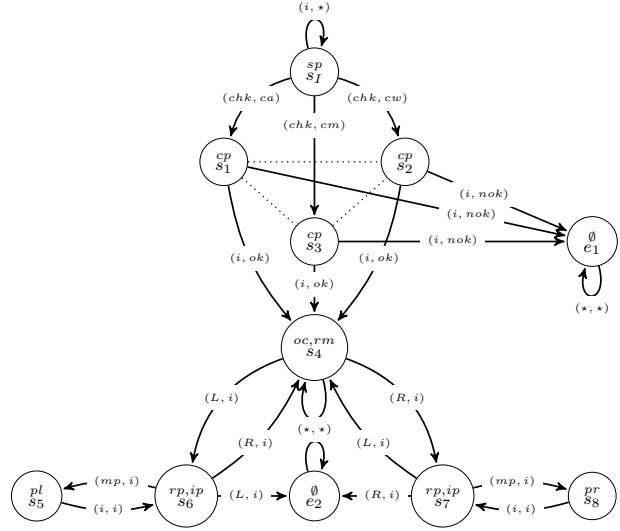

\section{Our procedure}\label{sec:procedure}
In this section, 
we provide a partial procedure to decide games with imperfect information and perfect recall strategies, a problem in general undecidable. 
%
To do this, given a model $M$ and a formula $\varphi$ in ATL$\!^*$, we need of the following three steps:

\begin{enumerate}
	\item to find the sub-models of $M$ in which there is perfect information;
	\item to check over the above sub-models the sub-formulas $\varphi'$ of $\varphi$;
	\item to use CTL$\!^*$ model checking to check the temporal remaining part $\psi$ of $\varphi$.
\end{enumerate}

Before providing these three points, we need to introduce a preprocessing procedure.

\subsection{Phase 0: preprocessing}\label{sec:prepro}

Here, we present the phase zero of our procedure.
%
%
Informally, given an iCGS $M$ and an ATL$\!^*$ formula $\varphi$, Algorithm \ref{alg:prepro} rewrites $\varphi$ in its equivalent Negative Normal Form (NNF), replaces the negated atoms by additional atoms, and updates the model $M$ accordingly.
With more detail, first, the NNF of $\varphi$ is generated and all the negated atoms are extracted (lines 1-2). Then, for each negated atom $\neg p$, the algorithm generates a new atomic proposition $np$ (line 4), add $np$ to the set of atomic propositions $AP$ (line 5), updates the formula (line 6), and updates the labeling function of $M$. For the latter, for each state $s$ of $M$, the algorithm checks if $p$ is false in $s$; if this is the case, it adds $np$ to the set of atomic propositions true on $s$ (lines 7-9). 

\begin{algorithm}
	\footnotesize
	\caption{$Preprocessing$ ($M$, $\varphi$)}
	\label{alg:prepro}
	\begin{algorithmic}[1] 
		\State $\varphi = NNF(\varphi)$;
		\State $neg$-$atoms$ = extract negated atoms from $\varphi$;
		\For{$\neg p \in neg$-$atoms$}
		\State generate atom $np$;
		\State $AP = AP \cup np$;
		\State replace $\neg p$ in $\varphi$ with $np$;
		\For{$s \in S$} 
		\If {$p \not\in V(s)$}
		\State $V(s) = V(s) \cup np$;
		\EndIf
		\EndFor
		\EndFor
	\end{algorithmic}
\end{algorithm}

\subsection{Phase 1: find sub-models}\label{sec:submodels}

Here, we present how to find the sub-models with perfect information inside the original model.
Before illustrating the algorithm we give the formal definitions of the classes of sub-models we analyze.
\begin{definition}[Negative sub-model]\label{subneg}
	Given an iCGS $M = \langle \Ag, AP, S, s_I,  \{Act_i\}_{i \in \Ag}, \allowbreak \{\sim_i\}_{i \in \Ag}, d, \delta, V \rangle$, we denote with $M' = \langle \Ag, AP, S', s'_I, \{Act_i\}_{i \in \Ag}, \{\sim'_i\}_{i \in \Ag}, d', \allowbreak \delta',  V' \rangle$ a negative sub-model of $M$, formally $M' \subseteq M$, such that:
	\begin{itemize}
		\item the set of states is defined as $S' = S^\star \cup \{s_\bot\}$, where $S^\star \subseteq S$, and $s'_I \in S^\star$ is the initial state.
		\item $\sim'_i$ is defined as the corresponding $\sim_i$ restricted to $S^\star$.
		\item The protocol function is defined as $d' : \Ag \times S' \rightarrow (2^{Act}\setminus\emptyset)$, where $d'(i,s) = d(i,s)$, for every $s \in S^\star$ and $d'(i,s_\bot) = Act_i$, for all $i \in \Ag$.
		\item The transition function is defined as $\delta' : S' \times ACT \rightarrow S'$, where given a transition $\delta(s, \vec{a}) = s'$, if $s, s' \in S^\star$ then $\delta'(s, \vec{a}) = \delta(s, \vec{a}) = s'$ else if $s' \in S \setminus S^\star$ and $s \in S'$ then $\delta'(s, \vec{a}) = s_\bot$. 
		\item for all $s\in S^\star$, $V'(s) = V(s)$ and $V'(s_\bot) = \emptyset$.
	\end{itemize}
\end{definition}

\begin{definition}[Positive sub-model]\label{subpos}
	Given an iCGS $M = \langle \Ag, AP, S, s_I,  \{Act_i\}_{i \in \Ag}, \allowbreak \{\sim_i\}_{i \in \Ag}, d, \delta, V \rangle$, we denote with $M' = \langle \Ag, AP, S', s'_I, \{Act_i\}_{i \in \Ag}, \{\sim'_i\}_{i \in \Ag}, d', \allowbreak \delta', V' \rangle$ a positive sub-model of $M$, formally $M' \subseteq M$, such that:
	\begin{itemize}
		\item the set of states is defined as $S' = S^\star \cup \{s_\top\}$, where $S^\star \subseteq S$, and $s'_I \in S^\star$ is the initial state.
		\item $\sim'_i$ is defined as the corresponding $\sim_i$ restricted to $S^\star$.
		\item The protocol function is defined as $d' : \Ag \times S' \rightarrow (2^{Act}\setminus\emptyset)$, where $d'(i,s) = d(i,s)$, for every $s \in S^\star$ and $d'(i,s_\top) = Act_i$, for all $i \in \Ag$.
		\item The transition function is defined as $\delta' : S' \times ACT \rightarrow S'$, where given a transition $\delta(s, \vec{a}) = s'$, if $s, s' \in S^\star$ then $\delta'(s, \vec{a}) = \delta(s, \vec{a}) = s'$ else if $s' \in S \setminus S^\star$ and $s \in S'$ then $\delta'(s, \vec{a}) = s_\top$. 
		\item for all $s\in S^\star$, $V'(s) = V(s)$ and $V'(s_\top) = AP$.
	\end{itemize}
\end{definition}
Note that, the above sub-models are still iCGSs.

The procedure to generate the set of positive and negative sub-models of $M$ with perfect information for the agents involved in the coalitions of $\varphi$ is presented in Algorithm \ref{alg:submodels}. 
Informally, the algorithm takes an iCGS $M$ and an ATL$\!^*$ formula $\varphi$.
The algorithm works with two sets. The set $substates$ contains the sets of states of $M$ that have to be evaluated and the set $candidates$ contains couples of sub-models of $M$ with perfect information. Each couple represents the positive and negative sub-models for a given subset of states of $M$.
In the first part of the algorithm (lines 1-4) the set of states and the indistinguishability relation are extracted from the model $M$ and the lists $substates$ and $candidates$ are initialised to $S$ and $\emptyset$, respectively. The main loop (lines 5-15) works until the list $substates$ is empty. In each iteration an element $\bar{S}$ is extracted from $substates$ and if there is an equivalence relation over the states in $\bar{S}$ for some agent $i \in Ag(\varphi)$\footnote{$Ag(\varphi)$ is the set of agents involved in the coalitions of the formula $\varphi$.} then two new subsets of states are generated to remove this relation (lines 8-11), otherwise $\bar{S}$ is a set of states with perfect information for the agents in $Ag(\varphi)$. So, from $\bar{S}$ two models are generated, a positive sub-model as described in Definition \ref{subpos} and a negative sub-model as described in Definition \ref{subneg}, and the couple is added to the list $candidates$. The latter is returned at the end of Algorithm~\ref{alg:submodels} (line 16).

\begin{algorithm}
	\footnotesize
	\caption{$FindSub$-$models$ ($M$, $\varphi$)}
	\label{alg:submodels}
	\begin{algorithmic}[1] 
		\State extract $\sim$ from $M$;
		\State extract $S$ from $M$;
		\State $candidates = \emptyset$;
		\State $substates =  \{S\}$;
		\While {$substates$ is not empty}
		\State extract $\bar{S}$ from $substates$;
		\If {there exists $s \sim_i s'$ with $s \neq s'$ and $i \in Ag(\varphi)$}
		\State $S_1 = \bar{S} \setminus \{s\}$;
		\State $substates = S_1 \cup substates$;
		\State $S_2 = \bar{S} \setminus\{s'\}$;
		\State $substates = S_2 \cup substates$;
		\Else
		\State $M_n = GenerateNegative(M, \bar{S})$;
		\State $M_p = GeneratePositive(M, \bar{S})$;
		\State $candidates = \langle M_n, M_p \rangle \cup candidates$;
		\EndIf
		\EndWhile
		\State \Return $candidates$;
	\end{algorithmic}
\end{algorithm}

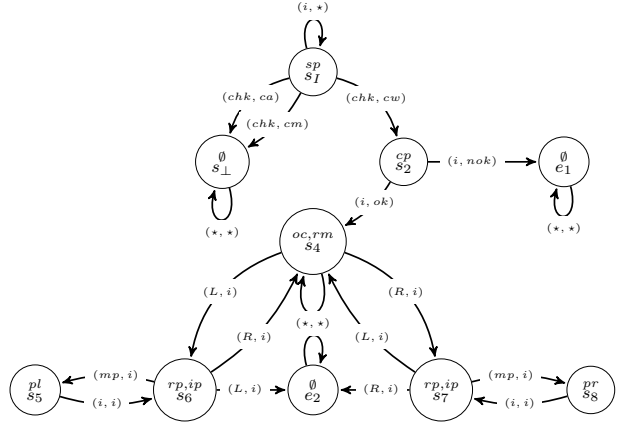
\begin{figure}[t]\centering
	\scalebox{0.80}{
		\begin{tikzpicture}
		\node [cnode]
		(S0)
		{\footnotesize $\stackrel{sp}{s_I}$};
		\node [cnode]
		(Sink)
		[below left of = S0, node distance = 6em]
		{\footnotesize $\stackrel{\emptyset}{s_\bot}$};
		\node [cnode]
		(S2)
		[below right of = S0, node distance = 6em]
		{\footnotesize $\stackrel{cp}{s_2}$};
		\node [cnode]
		(S4)
		[below of = S0, node distance = 8em]
		{\footnotesize $\stackrel{oc,rm}{s_4}$};
		\node [cnode]
		(E2)
		[below of = S4, node distance = 7em]
		{\footnotesize $\stackrel{\emptyset}{e_2}$};
		\node [cnode]
		(S6)
		[left of = E2, node distance = 6em]
		{\footnotesize $\stackrel{rp,ip}{s_6}$};
		\node [cnode]
		(S7)
		[right of = E2, node distance = 6em]
		{\footnotesize $\stackrel{rp,ip}{s_7}$};
		\node [cnode]
		(S5)
		[left of = S6, node distance = 7em]
		{\footnotesize $\stackrel{pl}{s_5}$};
		\node [cnode]
		(S8)
		[right of = S7, node distance = 7em]
		{\footnotesize $\stackrel{pr}{s_8}$};
		\node [cnode]
		(E1)
		[right of = S2, node distance = 7.5em]
		{\footnotesize $\stackrel{\emptyset}{e_1}$};

		\path[-stealth']
		(S4)  edge  [pos = 0.5, bend right = 30]
		node [above = -0.2] {\tiny $(L,i)$}
		(S6)
		edge  [pos = 0.5, bend left = 30]
		node [above = -0.2] {\tiny $(R,i)$}
		(S7)
		edge	[pos = 0.5, loop below]
		node [] {\tiny $(i,i)$}
		()
		
		(S0)
		edge    [pos = 0.5, bend right = 30]
		node    [above = -0.2] {\tiny $(chk,ca)$}
		(Sink)
		edge    [pos = 0.5, bend left = 15]
		node    [above = -0.2] {\tiny $(chk,cm)$}
		(Sink)
		edge    [pos = 0.5, bend left = 30]
		node    [above = -0.2] {\tiny $(chk,cw)$}
		(S2)
		edge	[pos = 0.5, loop above]
		node [] {\tiny $(i,\star)$}
		()
		
		(Sink)
		edge	[pos = 0.5, loop below]
		node [] {\tiny $(\star,\star)$}
		()
		
		(S2)
		edge    [pos = 0.4, bend left = 15]
		node    [above = -0.2] {\tiny $(i,ok)$}
		(S4)
		edge    [pos = 0.4]
		node    [above = -0.2] {\tiny $(i,nok)$}
		(E1)
		
		(S6)	
		edge	[pos = 0.4]
		node [above = -0.2] {\tiny $(L,i)$}
		(E2)
		edge	[pos = 0.4, bend right = 15]
		node [above = -0.2] {\tiny $(mp,i)$}
		(S5)
		edge	[pos = 0.4, bend right = 15]
		node [above = -0.2] {\tiny $(R,i)$}
		(S4)
		
		(E2)	
		edge	[pos = 0.5, loop above]
		node [] {\tiny $(\star,\star)$}
		()
		
		(E1)	
		edge	[pos = 0.5, loop below]
		node [] {\tiny $(\star,\star)$}
		()
		
		(S7)	edge	[pos = 0.4]
        node [above = -0.2] {\tiny $(R,i)$}
		(E2)
		edge	[pos = 0.4, bend left = 15]
		node [above = -0.2] {\tiny $(mp,i)$}
		(S8)
		edge	[pos = 0.4, bend left = 15]
		node [above = -0.2] {\tiny $(L,i)$}
		(S4)
		
		(S5)
		edge	[pos = 0.5, bend right = 15]
		node [above = -0.2] {\tiny $(i,i)$}
		(S6)

		(S8)
		edge	[pos = 0.5, bend left = 15]
		node [above = -0.2] {\tiny $(i,i)$}
		(S7)

		;

		\end{tikzpicture}
	} 
	 \caption{A negative sub-model generated by Algorithm \ref{alg:submodels} from model $M$ depicted in Figure \ref{fig:curiosity}.} \label{fig:ex-perfect-info} 
\end{figure}

\begin{example}
	In Figure~\ref{fig:ex-perfect-info}, we show a candidate negative sub-model $M_n$ extracted from the rover mission (see the model in Figure~\ref{fig:curiosity}) using Algorithm~\ref{alg:submodels}. 
	Note that, we can generate the positive sub-model counterpart by replacing the state $s_\bot$ with $s_\top$ and following the labelling rule described in Definition \ref{subpos}.
\end{example}

To conclude this part, we provide the complexity of the described procedure.

\begin{theorem}\label{thimperfect}
	Algorithm \ref{alg:submodels} terminates in $EXPTIME$.  
\end{theorem}

\begin{proof}
	The procedures $GenerateNegative()$ and $GeneratePositive()$ to produce models $M_n$ and $M_p$ are polynomial in the size of the original model $M$.
	In the worst case the loop to generate the list $candidates$ (lines 5-15) terminates in an exponential number of steps. Note that, the worst case happens when there is full indistinguishability, i.e. $\forall s, s' \in S$: $s \sim_i s'$. For the above reasoning, the result follows.
\end{proof}

\subsection{Phase 2: check sub-formulas}

Given a candidate, we present here how to check the sub-formulas on it.
The procedure is presented in Algorithm~\ref{alg:check}. 
Informally, the algorithm takes a candidate and an ATL$\!^*$ formula $\varphi$. 
It works with the set $result$ that contains tuples $\langle s, \psi, vatom_\psi\rangle$, where $s$ is a state of $S_p \cap S_n$\footnote{Note that as they are defined, $S_p$ and $S_n$ differ only for states $s_\top$ and $s_\bot$. So, $S_p \cap S_n = S_p \setminus \{s_\top\} = S_n \setminus \{s_\bot\}$.}, $\psi$ is a sub-formula of $\varphi$, and $vatom_\psi$ is the atomic proposition generated from $\psi$, where $v \in \{p,n\}$ is the type of the sub-model (\textit{i.e.}, positive or negative).
In line 1 the list $result$ is initialized to $\emptyset$. In line 2 the list $subformulas$ is initialized via the subroutine $SubFormulas()$. By $SubFormulas(\varphi)$ we take all the sub-formulas of $\varphi$ with only one strategic operator.
With more detail, we can see our formula $\varphi$ as a tree, where the root is $\varphi$, the other nodes are the sub-formulas of $\varphi$, and the leafs of the tree are the sub-formulas of $\varphi$ having only one strategic operator. Given this tree, the algorithm generates for each node $\psi$ an atomic proposition $atom_\psi$ and updates the tree by replacing every occurrence of $\psi$ in its ancestors with $atom_\psi$. Then, the nodes of the tree are stored in the list $subformulas$ by using a depth first search.
For instance, let us consider the formula $\varphi = \naww{\Gamma_1} X \all{\Gamma_2} p U r \land (\naww{\Gamma_3} X p \vee \all{\Gamma_4} q R p)$; when we apply $SubFormulas(\varphi)$ the tree in Figure \ref{fig:tree} is generated and we obtain the ordered list  $subformulas = \{\psi_2, \psi_1, \psi_3, \psi_4\}$. 
The extraction of the set of sub-formulas from $\varphi$ is necessary to verify whether there are sub-models satisfying/unsatisfying at least a part of $\varphi$ (not exclusively the whole $\varphi$). If that is the case, such sub-models can later on be used to verify $\varphi$ over the entire model $M$ through CTL$^*$ model checking.
The main loop (lines 3-12) works until all the sub-formulas of $\varphi$ are treated.
By starting from the first formula of the list, the loop in lines 5-11 proceeds for each state, and checks the current sub-formula against the currently selected sub-models $M_n$ and $M_p$.
Note that, in lines 6 and 9, we have $IR$ as verification mode\footnote{As usual in the verification process, we denote imperfect recall with $r$, perfect recall with $R$, imperfect information with $i$, and perfect information with $I$.}. Thus, the models are checked under the assumptions of perfect information and perfect recall. If the sub-formula is satisfied over the model $M_n$ (line 6), then by calling function $UpdateModel(M_n, s, atom_{\psi})$, the procedure updates the model $M_n$ by updating the set of atomic propositions as $AP = AP \cup atom_{\psi}$ and the labelling function of $s$ as $V(s) = V(s) \cup atom_{\psi}$ (line 7). Furthermore, in line 8, the list $result$ is updated by adding a new tuple.
The reasoning applied to the model $M_n$ is also applied to $M_p$ (lines 9-11).
At the end of the procedure, the state $s_\top$ is updated (line 5) since, as described in Definition \ref{subpos}, it needs to contain all the possible atoms that are in $M$.

\begin{figure}
	\begin{center}
		\large
		\mbox{\scalebox{0.38}[0.38]{
				\begin{tikzpicture}
					\node []
					(S0)
					{\huge $\stackrel{atom_{\psi_1} \land (atom_{\psi_3} \vee atom_{\psi_4})}{}$};
					\node [node distance = 6em]
					(S3)
					[below of = S0]
					{\huge $\stackrel{}{}$};
					\node [node distance = 17em]
					(S1)
					[left of = S3]
					{\huge $\stackrel{\psi_1 = \naww{\Gamma_1} X atom_{\psi_2}}{}$};
					\node [node distance = 14em]
					(S2)
					[right  of = S3]
					{\huge $\stackrel{atom_{\psi_3} \vee atom_{\psi_4}}{}$};
					\node [node distance = 5em]
					(S4)
					[below of = S1]
					{\huge $\stackrel{\psi_2 = \all{\Gamma_2} p U r}{}$};
					\node [node distance = 5em]
					(S7)
					[below of = S2]
					{};
					\node [node distance = 4em]
					(S8)
					[left of = S7]
					{\huge $\stackrel{\psi_3 = \naww{\Gamma_3} X p}{}$};
					\node [node distance = 4em]
					(S9)
					[right  of = S7]
					{\huge $\stackrel{\psi_4 = \all{\Gamma_4} q R p}{}$};
					
					\path[]
					(S0)  
					edge  (S1)
					edge  (S2)
					(S1)  
					edge  (S4)
					(S2)  
					edge  (S8)
					edge  (S9)
					;
				\end{tikzpicture}
			}
		}
		\caption{\label{fig:tree} The tree for the formula $\varphi_v$.
		}
	\end{center}
\end{figure}
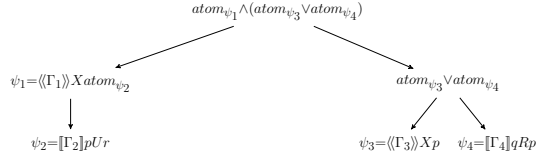

\begin{algorithm}
	\footnotesize
	\caption{$CheckSub$-$\!formulas$ ($\langle M_n, M_p \rangle$, $\varphi$)}
	\label{alg:check}
	\begin{algorithmic}[1] 
		\State $result = \emptyset$;
		\State $subformulas = SubFormulas(\varphi)$;
		\While {$subformulas \neq \emptyset$}
		\State extract $\psi$ from $subformulas$;
		\For{$s \in S_p \cap S_n$}
		\If {$M_n, s \models_{IR} \psi$}
		\State $M_n \!=\! U\!pdateM\!odel(M_n, s, atom_{\psi})$;
		\State $result = \langle s, \psi, natom_{\psi}\rangle \cup result$;
		\EndIf
		\If {$M_p, s \models_{IR} \psi$}
		\State $M_p \!=\! U\!pdateM\!odel(M_p, s, atom_{\psi})$;
		\State $result = \langle s, \psi, patom_{\psi}\rangle \cup result$;
		\EndIf
		\EndFor
		\State $M_p \!=\! U\!pdateM\!odel(M_p, s_\top, atom_{\psi})$;
		\EndWhile
		\State \Return $result$;
	\end{algorithmic}
\end{algorithm}

\begin{example}	
Given the negative sub-model $M_n$ as in Figure~\ref{fig:ex-perfect-info} and its positive counterpart $M_p$ as input of  Algorithm~\ref{alg:check} we can analyze formulas $\varphi_1 = \naww{rover} F ((oc \land rm) \land  F ((pl \vee pr) \land F (oc \land rm)))$, $\varphi_2 = \naww{rover, mechanic} F ((oc \land rm) \land  \naww{rover} F ((pl \vee pr) \land F (oc \land rm)))$, and $\varphi'_3 = \naww{rover, mechanic} F (rp \land nip \land \naww{rover} F ( (pl \vee pr) \land F (oc \land rm)))$\footnote{$\varphi_3$ described in Section \ref{example:rover} is substituted with $\varphi'_3$ where $\neg ip$ is replaced with $nip$ as described in the preprocessing algorithm of Section \ref{sec:prepro}.}.
We start with $\varphi_1$ where the set of sub-formulas is initialized to the singleton $\varphi_1 = \naww{rover} F ((oc \land rm) \land  F ((pl \vee pr) \land F (oc \land rm)))$ since there is only one strategic operator in $\varphi_1$.
By line 4, the algorithm extracts $\varphi_1$ and by the loop in lines 5-11 it verifies that the formula is true in the negative and positive sub-models in states $s_4$, $s_5$, $s_6$, $s_7$, and $s_8$. By line 7 (resp., 10) it updates the model $M_n$ (resp., $M_p$) by adding $atom_{\psi_1}$ in $AP$ and by updating the labeling function as $L'(s_i) = L(s_i) \cup \{atom_{\psi_1}\}$, for all $i \in \{4,5,6,7,8\}$. Since there is only one strategic operator, the procedure is concluded.
For $\varphi_2$, the set $subformulas$ via $SubFormulas()$ is initialized to the list $\{\psi_1 = \naww{rover} F ((pl \vee pr) \land F (oc \land rm))$, $\psi_2 = \naww{rover, mechanic} F ((oc \land rm) \land atom_{\psi_1})\}$.
By line 4, the algorithm extracts $\psi_1$ and by the loop in lines 5-11 it verifies that the formula is true in the negative and positive sub-models in states $s_4$, $s_5$, $s_6$, $s_7$, and $s_8$. By line 7 (resp., 10) it updates the model $M_n$ (resp., $M_p$) by adding $atom_{\psi_1}$ in $AP$ and by updating the labeling function as $L'(s_i) = L(s_i) \cup \{atom_{\psi_1}\}$, for all $i \in \{4,5,6,7,8\}$.
For the second iteration of the main loop (lines 3-11) the formula $\psi_2$ is analyzed.
In this case, $\psi_2$ is true in the negative and positive sub-models in states $s_I$, $s_2$, $s_4$, $s_5$, $s_6$, $s_7$, and $s_8$. Again, the models $M_n$ and $M_p$ are updated by adding $atom_{\psi_2}$ in $AP$ and by updating the labeling function as $L'(s_i) = L(s_i) \cup \{atom_{\psi_1}\}$, for all $i \in \{I,2,4,5,6,7,8\}$.
For $\varphi_3$, the set of sub-formulas is initialized to the list $\{\psi_1 = \naww{rover} F ( (pl \vee pr) \land F (oc \land rm))$, $\psi_2 = \naww{rover, mechanic} F (rp \land nip \land atom_{\psi_1})\}$.
Since $\psi_1$ is the same as before for $\varphi_2$ we have the same behavior of the algorithm.
For the second iteration, $\psi_2$ is false in all the states in $S_p \cap S_n$. So, no update are done in the list $result$.
\end{example}

To conclude this part, we provide the complexity result.

\begin{theorem}\label{th:check}
	Algorithm \ref{alg:check} terminates in 2$EXPTIME$.  
\end{theorem}

\begin{proof} 
	The procedure $SubFormulas()$ is polynomial in the size of the formula under exam. The procedure $UpdateModel()$ is polynomial in the size of the model. ATL$^*$ model checking, called in lines 6 and 9, is 2$EXPTIME$ \cite{AlurHenzingerKupferman97}.
	Furthermore, the number of iterations of the inner part (lines 5-11) is $|\varphi|\cdot|S|$, i.e. polynomial in the size of the model and of the formula.
	Then, the total complexity derives from ATL$^*$ model checking and is 2$EXPTIME$.
\end{proof}

\subsection{Phase 3: CTL$^*$ verification}\label{subrun}

In the previous sections, we showed how to extract sub-models of a model $M$ satisfying/unsatisfying at least one sub-formula $\varphi'$ of $\varphi$. Here, we present in Algorithm~\ref{alg:CTL} how to apply CTL$\!^*$ model checking to try to conclude the satisfaction/violation of $\varphi$ in $M$ by using such sub-models. 
%
%
The algorithm takes in input an iCGS $M$, an ATL$\!^*$ formula $\varphi$, and a set $result$. 
%
The algorithm starts by updating the model $M$ with the atoms generated from Algorithm \ref{alg:check} and stored in $result$ (lines 2-4). 
%
In lines 5-11, the procedure generates formulas $\psi_p$ and $\psi_n$ in accordance with the tuples in $result$. 
%
In this part of the procedure, the algorithm uses the subroutine $UpdateFormula(\varphi,\psi,atom_{\psi})$. The latter modifies the formula $\varphi$ by replacing each occurrence of $\psi$ with $atom_{\psi}$, where $atom_{\psi}$ is a unique atom identifier for formula $\psi$.  As an example, let us consider $\varphi = \naww{\Gamma} X \naww{\Gamma} p U q$, with $\psi = \naww{\Gamma} p U q$. When we apply $UpdateFormula(\varphi,\psi,atom_{\psi})$, we obtain an updated version of $\varphi$, where all occurrences of $\psi$ have been replaced by $atom_{\psi}$, \textit{i.e.} $\varphi$ becomes $\naww{\Gamma} X atom_{\psi}$.
The formulas $\psi_n$ and $\psi_p$ represent the ATL$\!^*$ formula which remains to verify in $M$. To achieve this, we transform $\psi_n$ and $\psi_p$ into their CTL$\!^*$ counterparts $\varphi_A$ and $\varphi_E$, respectively (lines 12-13).
Given a formula $\varphi$ in ATL$\!^*$, we denote with $\varphi_A$ (resp., $\varphi_E$) the universal formula (resp., existential) of $\varphi$, \textit{i.e.} we substitute each occurrence of a strategic operator in $\varphi$ with the universal (resp., existential) operator of CTL$\!^*$. Given $\varphi_A$ and $\varphi_E$ we can perform standard CTL$\!^*$ model checking.
First, by verifying if $\varphi_A$ is satisfied by $M$ (line 14). If this is the case 
we can derive that $M\models\varphi$ and set $k = \top$. Otherwise, the algorithm continues verifying if $\varphi_E$ is not satisfied by $M$ (line 16). If this is the case 
we can derive that $M\!\not\models\!\varphi$ and set $k \!=\! \bot$.  
If neither of the two previous conditions are satisfied then the inconclusive result (\textit{i.e.}, $?$) is returned.

\begin{algorithm}
	\footnotesize
	\caption{$Verification$ ($M$, $\varphi$, $result$)}
	\label{alg:CTL}
	\begin{algorithmic}[1] 
		\State $k$ = $?$;
		\For{$s \in S$}
		\State take set $atoms$ from $result(s)$;
		\State $UpdateModel$($M$, $s$, $atoms$);	
		\EndFor
		\State $\psi_n = \varphi$, $\psi_p = \varphi$; 
		\While{$result$ is not empty}
		\State extract $\langle s, \psi, vatom_{\psi}\rangle$ from $result$; 
		\If{$v = n$}
		\State $\psi_n$ = $UpdateFormula$($\psi_n$, $\psi$, $natom_{\psi}$);
		\Else
		\State $\psi_p$ = $UpdateFormula$($\psi_p$, $\psi$, $patom_{\psi}$);
		\EndIf
		\EndWhile
		\State $\varphi_A$ = $FromATLtoCTL$($\psi_n$, $n$);
		\State $\varphi_E$ = $FromATLtoCTL$($\psi_p$, $p$);
		\If {$M\models\varphi_A$}
		\State $k$ = $\top$;
		\EndIf
		\If {$M\not\models\varphi_E$}
		\State $k$ = $\bot$;
		\EndIf
		\State \Return $k$;
	\end{algorithmic}
\end{algorithm} 

\begin{example}
As for the previous example, we continue to analyze formulas $\varphi_1$, $\varphi_2$, and $\varphi_3$.
For $\varphi_1$ the list $result$ is composed by the tuples: $\langle s_4, \psi_1, vatom_{\psi_1}\rangle, \langle s_5, \psi_1, \allowbreak vatom_{\psi_1}\rangle, \langle s_6, \psi_1, vatom_{\psi_1}\rangle, \langle s_7, \psi_1, vatom_{\psi_1}\rangle, \textit{ and } \langle s_8, \psi_1, vatom_{\psi_1}\rangle$ where $v \in \{n,p\}$. In lines 2-4, the model $M$ depicted in Figure \ref{fig:curiosity} is updated with the atoms $natom_{\psi_1}$ and $patom_{\psi_1}$ in states $s_4$, $s_5$, $s_6$, $s_7$, and $s_8$. Then, in lines 6-11, $\psi_n$ and $\psi_p$ are generated as $natom_{\psi_1}$ and $patom_{\psi_1}$, respectively. Since there are no strategic operators in the latter formulas then $\varphi_A = \psi_n$ and $\varphi_E = \psi_p$. The condition in line 14 is not satisfied but it is in line 16. So, the output of the procedure is $\bot$, i.e. the formula is false in the model $M$. 
For $\varphi_2$ the list $result$ is composed by the tuples: $\langle s_4, \psi_1, vatom_{\psi_1}\rangle, \langle s_5, \psi_1, vatom_{\psi_1}\rangle, \langle s_6, \psi_1, vatom_{\psi_1}\rangle, \allowbreak \langle s_7, \psi_1, vatom_{\psi_1}\rangle,  \langle s_8, \psi_1, \allowbreak vatom_{\psi_1}\rangle, \langle s_I, \psi_2,  vatom_{\psi_2}\rangle, \langle s_2, \psi_2, vatom_{\psi_2}\rangle, \langle s_4, \psi_2, vatom_{\psi_2}\rangle, \langle s_5, \psi_2, vatom_{\psi_2} \allowbreak\rangle, \langle s_6, \psi_2, vatom_{\psi_2}\rangle, \langle s_7, \psi_2, vatom_{\psi_2}\rangle,  \textit{ and } \langle s_8, \psi_2, vatom_{\psi_2}\rangle$ where $v \in \{n,p\}$. In lines 2-4, the model $M$ depicted in Figure \ref{fig:curiosity} is updated with the atoms $natom_{\psi_1}$ and $patom_{\psi_1}$ in states $s_4$, $s_5$, $s_6$, $s_7$, and $s_8$; and atoms $natom_{\psi_2}$ and $patom_{\psi_2}$ in states $s_I$, $s_2$, $s_4$, $s_5$, $s_6$, $s_7$, and $s_8$. Then, in lines 6-11 $\psi_n$ and $\psi_p$ are generated as $natom_{\psi_2}$ and $patom_{\psi_2}$, respectively. Since there are no strategic operators in the latter formulas then $\varphi_A = \psi_n$ and $\varphi_E = \psi_p$. The condition in line 16 is not satisfied but it is in line 14. So, the output of the procedure is $\top$, i.e. the formula is true in the model $M$. 
For $\varphi_3$ the list result is composed by the tuples: $\langle s_4, \psi_1, vatom_{\psi_1}\rangle$, $\langle s_5, \psi_1, vatom_{\psi_1}\rangle$, $\langle s_6, \psi_1, vatom_{\psi_1}\rangle$, $\langle s_7, \psi_1, vatom_{\psi_1}\rangle$,  and $\langle s_8, \psi_1, vatom_{\psi_1}\rangle$ where $v \in \{n,p\}$. In lines 2-4, the model $M$ depicted in Figure \ref{fig:curiosity} is updated with the atoms $natom_{\psi_1}$ and $patom_{\psi_1}$ in states $s_4$, $s_5$, $s_6$, $s_7$, and $s_8$. Then, in lines 6-11 $\psi_n$ and $\psi_p$ are generated as $\naww{rover, mechanic} F (rp \land nip \land natom_{\psi_1})$ and $\naww{rover, mechanic} F (rp \land nip \land patom_{\psi_1})$, respectively. Since there are strategic operators in the latter formulas then by lines 12-13 the algorithm replaces the strategic operators with path quantifiers. In particular, the algorithm produces formulas $\varphi_A = A F (rp \land nip \land natom_{\psi_1})$ and $\varphi_E = E F (rp \land nip \land patom_{\psi_1})$. The condition in line 14 is not satisfied but it is in line 16. So, the output of the procedure is $\bot$, i.e. the formula is false in the model $M$. 

Note that, in these examples, with our algorithm we can solve problems that are in the class of undecidable problems. Furthermore, we preserve the truth values, i.e. they are the same as described in Section \ref{example:rover} in the context of imperfect information and perfect recall strategies. What remains is the correctness of our algorithm that we discuss in Section \ref{soundness}.

\end{example}

\begin{theorem}\label{thCTL}
Algorithm~\ref{alg:CTL} terminates in $PSPACE$.
\end{theorem}
\begin{proof}
The cost of the loop in lines 2-4 is polynomial in the size of the model. Since the complexity of $UpdateFormula()$ is polynomial in the size of the formula, then the cost of the loop in lines 6-11 is polynomial in the size of the formula and in the size of the model. 
The procedure $FromATLtoCTL()$ is polynomial in the size of the formula. The CTL$^*$ model checking problem is $PSPACE$ (lines 14 and 16). Then, the total complexity is $PSPACE$.
\end{proof} 

\subsection{The overall model checking procedure}\label{wholeprocedure}
Given Algorithms \ref{alg:prepro}, \ref{alg:submodels}, \ref{alg:check}, and \ref{alg:CTL}, we can provide the overall procedure as described in Algorithm \ref{alg:full}.

\begin{algorithm}
	\footnotesize
	\caption{$ModelCheckingProcedure$ ($M$, $\varphi$)}
	\label{alg:full}
	\begin{algorithmic}[1] 
		\State $Preprocessing(M, \varphi)$;
		\State $candidates$ = $FindSub$-$models$$(M, \varphi)$;
		\While {$candidates$ is not empty}
		\State extract $\langle M_n, M_p \rangle$ from $candidates$;
		\State $result$ = $CheckSub$-$\!formulas$$(\langle M_n, M_p \rangle,\varphi)$;
		\State $k$ = $Verification$$(M, \varphi, result)$; 
		\If {$k \neq ?$}
		\State \Return $k$;
		\EndIf
		\EndWhile
		\State \Return $?$;
	\end{algorithmic}
\end{algorithm}

The $ModelCheckingProcedure()$ takes in input a model $M$ and a formula $\varphi$, and calls the function $Preprocessing()$ to generate the NNF of $\varphi$ and to replace all negated atoms with new positive atoms inside $M$ and $\varphi$. After that, it calls the function $FindSub$-$models()$ to generate all the positive and negative sub-models that represent all the possible sub-models with perfect information. Then, there is a while loop (lines 3-8) that for each candidate checks the sub-formulas true on the sub-models via $CheckSub$-$formulas()$ and the truth value of the whole formula via $Verification()$. If the output of the latter procedure is different from the unknown then the result is directly returned (line 8).

\begin{theorem}\label{theo:full}
	Algorithm~\ref{alg:full} terminates in $2EXPTIME$.
\end{theorem}

\begin{proof}
	In the worst case the while loop in lines 3-8 needs to check all the candidates (i.e. the condition in line 7 is always false) and the size of the list of candidates is equal to the size of the set of states of $M$ (i.e. polynomial in the size of $M$). So, the total complexity is still determined by the subroutines and follows by Theorems \ref{thimperfect}, \ref{th:check}, and \ref{thCTL}. 
\end{proof}

\subsection{Soundness}\label{soundness}	

Before showing the soundness of our algorithm, we need to provide some auxiliary lemmas. Since in our procedure we use ATL$^*$ formulas in NNF, in the rest of the section we will consider this type of formulas.
We start with a preservation result from negative sub-models to the original model.

\begin{lemma}\label{lemma:negsub}
	Given a model $M$, a negative sub-model with perfect information $M_n$ of $M$, and a formula $\varphi$ of the form $\varphi = \naww{A} \psi$ (resp., $\all{A} \psi$) for some $A \subseteq Ag$. For any $s \in S_n \setminus \{s_\bot\}$, we have that:
	\begin{center}
		$M_n, s \models \varphi \Rightarrow M, s \models \varphi$
	\end{center}
\end{lemma}

\begin{proof}
	By Definition \ref{satisfaction}, $M_n, s \models \varphi$ if and only if there exists a collective strategy $\sigma_A$ such that for all paths (resp., for all collective strategies $\sigma_A$ there exists a path) $p \in out(s,\sigma_A)$, $M_n, p \models \psi$. 
	By Definition \ref{subneg}, we know that states in $S_n$ are all but $s_\bot$ in $S$. So, we need to consider two classes of paths: 
	\begin{enumerate}
		\item $p$ where for all $i > 0$, $p_i \neq s_\bot$;
		\item $p$ where there exists an $i > 1$ such that $p_i = s_\bot$;
	\end{enumerate}
	For (1), we have that $p$ is also in $M$ and by consequence if $M_n, p \models \psi$ then $M, p \models \psi$.
	For (2), each path $p$ has the following structure $p = p^s \cdot p^\bot$, where $p^s$ is the prefix of $p$ where $p^s_i \in S$ for each $0 < i \leq |p^s|$, and $p^\bot$ is the suffix of $p$ where $p^\bot = s^\omega_\bot$. Now, 
	$M_n, p \models \psi$ if and only if $\psi$ is satisfied in the prefix $p^s$ but since $p^s$ is a prefix also in $M$ then $M, p^s \cdot p' \models \psi$, no matter what the states of the suffix $p'$ are.
\end{proof}

We also consider the preservation result from positive sub-models to the original model.

\begin{lemma}\label{lemma:possub}
	Given a model $M$, a positive sub-model with perfect information $M_p$ of $M$, and a formula $\varphi$ of the form $\varphi = \naww{A} \psi$ (resp., $\all{A} \psi$) for some $A \subseteq Ag$. For any $s \in S_p \setminus \{s_\top\}$, we have that:
	\begin{center}
		$M_p, s \not\models \varphi \Rightarrow M, s \not \models \varphi$ 
	\end{center} 
\end{lemma}

\begin{proof}
	By Definition \ref{satisfaction}, $M_p, s \not \models \varphi$ if and only for all collective strategies $\sigma_A$ there exists (resp., there exists a collective strategy $\sigma_A$ such that for all paths) $p \in out(s,\sigma_A)$, $M_p, p \not \models \psi$.
	By Definition \ref{subpos}, we know that states in $S_p$ are all but $s_\top$ in $S$. So, we need to consider two classes of paths: 
	\begin{enumerate}
		\item $p$ where for all $i > 0$, $p_i \neq s_\top$;
		\item $p$ where there exists an $i > 1$ such that $p_i = s_\top$;
	\end{enumerate}
	For (1), we have that $p$ is also in $M$ and by consequence if $M_p, p \not\models \psi$ then $M, p \not\models \psi$.
	For (2), each path $p$ as the following structure $p = p^s \cdot p^\top$, where $p^s$ is the prefix of $p$ where $p^s_i \in S$ for each $0 < i \leq |p^s|$, and $p^\top$ is the suffix of $p$ where $p^\top = s^\omega_\top$. Now, $M_p, p \not\models \psi$ if and only if $\psi$ is not satisfied in the prefix $p^s$ but since $p^s$ is also a prefix in $M$ then $M, p^s \cdot p' \not\models \psi$, no matter what the states of the suffix $p'$ are.
\end{proof}
	
Now, we present a preservation result from CTL$^*$ universal formulas to ATL$^*$ formulas.

\begin{lemma}\label{theo:UnivCTLimpliesATL}
	Given a model $M$, a formula $\varphi$ in ATL$\!^*$ written in NNF, and the CTL$^*$ universal version $\varphi_A$ of $\varphi$. 
	For any $s \in S$, we have that: 
	\begin{center}
		$M,s \models \varphi_A \Rightarrow M,s \models \varphi$
	\end{center} 
\end{lemma}

\begin{proof}
	First, consider the formula $\varphi = \naww{\Gamma} \psi$, in which $\Gamma \subseteq Ag$ and $\psi$ is a temporal formula without quantifications. By consequence, $\varphi_A = A \psi$. By the semantics of CTL$\!^*$, $M,s \models A \psi$ iff for all paths $p$ with $p_1 = s$, $M,p \models \psi$. As presented in~\cite{AHK02}, the formula $A \psi$ in CTL$\!^*$ is equivalent to $\naww{\emptyset} \psi$ in ATL$\!^*$. In fact, by Definition~\ref{satisfaction}, the set of paths that need to verify $\psi$ are in $out(s, \emptyset)$, \textit{i.e.} all the paths of $M$ from $s$. Then, we need to prove that: if $M,s \models \naww{\emptyset} \psi$ then $M,s \models \naww{\Gamma} \psi$. But, this follows directly by the semantics in Definition \ref{satisfaction}. 
	
	Now, consider the formula $\varphi = \all{\Gamma} \psi$, in which $\Gamma \subseteq Ag$ and $\psi$ is a temporal formula without quantifications. By consequence, $\varphi_A = A \psi$. By the semantics of CTL$\!^*$, $M,s \models A \psi$ iff for all paths $p$ with $p_1 = s$, $M,p \models \psi$. The formula $A \psi$ in CTL$\!^*$ is equivalent to $\all{Ag} \psi$ in ATL$\!^*$. In fact, by Definition~\ref{satisfaction}, for each joint strategy $\Sigma_{Ag}$ the resulted path needs to verify $\psi$. Since we consider all the possible joint strategies for the whole set of agents by consequence all the paths of $M$ from $s$ need to satisfy $\psi$. Then, we need to prove that: if $M,s \models \all{Ag} \psi$ then $M,s \models \all{\Gamma} \psi$. But, this follows directly by the semantics in Definition~\ref{satisfaction}. 
	
	To conclude the proof, note that if we have a formula with more strategic operators then we can use a classic bottom-up approach.
\end{proof}

Last but not least, we provide a preservation result from CTL$^*$ existential formulas to ATL$^*$ formulas.

\begin{lemma}\label{theo:ExistCTLimpliesATL}
	Given a model $M$, a formula $\varphi$ in ATL$\!^*$ written in NNF, and the CTL$^*$ existential version $\varphi_E$ of $\varphi$. For any $s \in S$, we have that: 
	\begin{center}
		$M,s \not \models \varphi_E \Rightarrow M,s \not \models \varphi$
	\end{center}	
\end{lemma}

\begin{proof}
	First, consider the formula $\varphi = \naww{\Gamma} \psi$, in which $\Gamma \subseteq Ag$ and $\psi$ is a temporal formula without quantifications. By consequence, $\varphi_E = E \psi$. By the semantics of CTL$^*$, $M,s \not \models E \psi$ iff there is not a paths $p$ with $p_1 = s$, such that $M,p \models \psi$. As presented in~\cite{AHK02}, the formula $E \psi$ in CTL$^*$ is equivalent to $\naww{Ag} \psi$ in ATL$^*$. In fact, by Definition~\ref{satisfaction}, we take $out(s, \Sigma_{Ag})$, \textit{i.e.} only one path of $M$. Then, we need to prove that: if $M,s \not \models \naww{Ag} \psi$ then $M,s \not \models \naww{\Gamma} \psi$. But, this follows directly by the semantics in Definition~\ref{satisfaction}. 
	
	Now, consider the formula $\varphi = \all{\Gamma} \psi$, in which $\Gamma \subseteq Ag$ and $\psi$ is a temporal formula without quantifications. By consequence, $\varphi_E = E \psi$. By the semantics of CTL$^*$, $M,s \not \models E \psi$ iff there is not a path $p$ with $p_1 = s$, such that $M,p \models \psi$. The formula $E \psi$ in CTL$^*$ is equivalent to $\all{\emptyset} \psi$ in ATL$^*$. In fact, by Definition~\ref{satisfaction}, we take $out(s, \emptyset)$, \textit{i.e.} all the paths of $M$ starting from $s$, and verify if there exists a path in $out(s, \emptyset)$ satisfying $\psi$. Then, we need to prove that: if $M,s \not \models \all{\emptyset} \psi$ then $M,s \not \models \all{\Gamma} \psi$. But, this follows directly by the semantics in Definition~\ref{satisfaction}.
	
	To conclude the proof, note that if we have a formula with more strategic operators then we can use a classic bottom-up approach.
\end{proof}

Finally, we have all ingredients to show our main result.

\begin{theorem}\label{theo:soundness}
	Algorithm~\ref{alg:full} is sound: if the value returned is different from $?$, then $M \models \varphi$ iff $k = \top$.
\end{theorem}

\begin{proof}
	
	Suppose that the value returned is different from $?$. In particular, either $k = \top$ or $k = \bot$. 
	If $M \models \varphi$ and $k = \bot$, then by Algorithm~\ref{alg:CTL} and~\ref{alg:full}, we have that $M' \not\models \varphi_E$.
	Now, there are two cases: (1) $M$ and $M'$ are the same models or (2) $M$ differs from $M'$ for some atomic propositions added to $M'$ in lines 2-4 of Algorithm \ref{alg:CTL}.
	In case (1), we know that $M$ and $M'$ are labeled with the same atomic propositions and thus $M' \not\models \varphi_E$ implies $M \not\models \varphi_E$ and, by Lemma \ref{theo:ExistCTLimpliesATL}, $M \not\models \varphi$, a contradiction. Hence, $k = \top$ as required. 
	In case (2), we know that $M$ and $M'$ differs for some atomic propositions. Suppose that $M'$ has only one additional atomic proposition $atom_\psi$, i.e. $AP' = AP \cup \{atom_\psi\}$. The latter means that Algorithm \ref{alg:check} found a positive sub-model $M_p$ in which $M_p,s \models \psi$, for some $s \in S_p$. By Lemma \ref{lemma:possub}, for all $s \in S_p$, we know that if $M_p,s \not \models \psi$ then $M,s \not\models \psi$. So, $M'$ over-approximates $M$, i.e. there could be some states that in $M'$ are labeled with $atom_\psi$ but they don't satisfy $\psi$ in $M$. Thus, if $M' \not\models \varphi_E$ then $M \not\models \varphi_E$ and, by Lemma \ref{theo:ExistCTLimpliesATL}, $M \not\models \varphi$, a contradiction. Hence, $k = \top$ as required. 
	Obviously, we can generalize the above reasoning in case $M$ and $M'$ differ for multiple atomic propositions.
	On the other hand, if $k = \top$ then by Algorithm~\ref{alg:CTL} and~\ref{alg:full}, we have that $M' \models \varphi_A$.
	Again, there are two cases: (1) $M$ and $M'$ are the same models or (2) $M$ differs from $M'$ for some atomic propositions added to $M'$ in lines 2-4 of Algorithm \ref{alg:CTL}.
	In case (1), we know that $M$ and $M'$ are labeled with the same atomic propositions and thus $M' \models \varphi_A$ implies $M \models \varphi_A$ and, by Lemma \ref{theo:UnivCTLimpliesATL}, $M \models \varphi$ as required. 
	In case (2), we know that $M$ and $M'$ differs for some atomic propositions. Suppose that $M'$ has only one additional atomic proposition $atom_\psi$, i.e. $AP' = AP \cup \{atom_\psi\}$. The latter means that Algorithm \ref{alg:check} found a negative sub-model $M_n$ in which $M_n,s \models \psi$, for some $s \in S_n$. By Lemma \ref{lemma:negsub}, for all $s \in S_n$, we know that if $M_n,s \models \psi$ then $M,s \models \psi$. So, $M'$ under-approximates $M$, i.e. there could be some states that in $M'$ are not labeled with $atom_\psi$ but they satisfy $\psi$ in $M$. Thus, if $M' \models \varphi_A$ then $M \models \varphi_A$ and, by Lemma \ref{theo:UnivCTLimpliesATL}, $M \models \varphi$ as required. 
	Obviously, we can generalize the above reasoning in case $M$ and $M'$ differ for multiple atomic propositions.
\end{proof}

\section{Our tool}\label{sec:tool}

The algorithms previously presented have been implemented in Java\footnote{ \url{https://github.com/AngeloFerrando/StrategyCTL}}. 
%
The tool expects a model in input formatted as a Json file. This file is then parsed, and an internal representation of the model and formula are generated. After the preprocessing phase, Algorithm~\ref{alg:submodels} is called to extract the sub-models of our interest. The verification of a sub-model against a sub-formula is achieved by translating the sub-model into its equivalent ISPL (Interpreted Systems Programming Language) program, which is then verified by using the model checker MCMAS\footnote{\url{https://vas.doc.ic.ac.uk/software/mcmas/}}\cite{LomuscioRaimondi06c}. This corresponds to the verification steps for Algorithm~\ref{alg:check} at lines 6 and 9, and Algorithm~\ref{alg:CTL} at lines 14 and 16.
The entire manipulation, from parsing the model formatted in Json, to translating the latter to its equivalent ISPL program, has been performed by extending an existent Java library \cite{StrategicTool}; the rest of the tool derives directly from the algorithms presented in this paper. 
%

\vspace{-0.5em}
\subsection{Experiments} 

We have tested our tool over the rover's mission, on a machine with the following specifications: Intel(R) Core(TM) i7-7700HQ CPU @ 2.80GHz, 4 cores 8 threads, 16 GB RAM DDR4.
%
By Considering the model of Figure~\ref{fig:curiosity} and the three formulas $\varphi_1$, $\varphi_2$ and $\varphi_3$ presented in the same section, Algorithm~\ref{alg:submodels} finds 3 sub-models with perfect information (precisely, 3 negative and 3 positive sub-models). Such sub-models are then evaluated by applying first Algorithm~\ref{alg:check}, and then Algorithm~\ref{alg:CTL} (as shown in Algorithm~\ref{alg:full}). The resulting implementation takes 0.2 [sec] to conclude the satisfaction of $\varphi_2$ and the violation of $\varphi_1$ and $\varphi_3$ on the model (empirically confirming our expectations). 
The example of Figure~\ref{fig:curiosity} is expressly small, since it is used as running example to help the reader to understand the contribution. However, the actual rover example we tested on our tool contains hundreds of states and dozens of agents (we have multiple rovers and mechanics in the real example). Such more complex scenario served us as a stressed test to analyse the performance of our tool for more realistic MAS. Since Algorithm~\ref{alg:full} is $2EXPTIME$, when the model grows, the performance are highly affected. Nonetheless, even when the model reaches 300 states, the tool concludes in $\sim$10 [min].
Other than the rover's mission, our tool has been further experimented on a large set of automatically and randomly generated iCGSs. The objective of these experiments have been to show how many times our algorithm has returned a conclusive verdict. For each model, we have run our procedure and counted the number of times a solution has been returned. Note that, our approach concludes in any case, but since the general problem is undecidable, the result might be inconclusive (i.e. $?$). In Figure~\ref{fig:experiments}, we report our results by varying the percentage of imperfect information (x axis) inside the iCGSs, from $0\%$ (perfect information, i.e. all states are distinguishable for all agents), to $100\%$ (no information, i.e. no state is distinguishable for any agent). For each percentage selected, we have generated $10000$ random iCGSs and counted the number of times our algorithm has returned with a conclusive result (i.e. $\top$ or $\bot$). As it can be seen in Figure~\ref{fig:experiments}, our algorithm has returned a conclusive verdict for almost the $80\%$ of the models analysed (y axis). It is important to notice how this result has not been influenced (empirically) by the percentage of imperfect information added inside the iCGSs. In fact, the accuracy of the algorithm is not determined by the number of indistinguishable states, but by the topology of the iCGS and the structure of the formula under exam. In order to have pseudo-realistic results, the automatically generated iCGSs have varied over the number of states, the complexity of the formula to analyse, and the number of transitions among states. More specifically, not all iCGSs have been generated completely randomly, but a subset of them has been generated considering more interesting topologies (similar to the rover's mission). This has contributed to have more realistic iCGSs, and consequently, more realistic results.  
The results obtained by our experiments using our procedure are encouraging. Unfortunately, no benchmark of existing iCGSs -- to test our tool on -- exists, thus these results may vary on more realistic scenarios. Nonetheless, considering the large set of iCGSs we have experimented on, we do not expect substantial differences.

\begin{figure}
	\centering
	\scalebox{0.8}{
		\includegraphics[width=\linewidth]{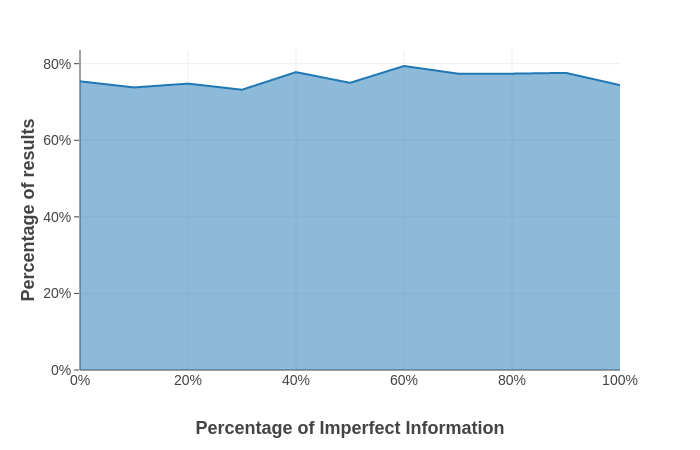}
	}
	\caption{Experimental results.}
	\label{fig:experiments}
\end{figure}

\section{Conclusions and Future work}


In this paper we have proposed a procedure to overcome the issue in employing logics for strategic reasoning in the context of MAS under perfect recall and imperfect information, a problem in general undecidable. Specifically, we have shown how to generate the sub-models in which the verification of strategic objectives is decidable and then used CTL$\!^*$ model checking to provide a verification result. We have shown how the entire procedure has the same complexity of ATL$\!^*_{IR}$ model checking. 
We have also implemented our procedure by using MCMAS and provided encouraging experimental results.
%

In future work, we intend to extend our procedure to increase the types of games and specifications that we can cover. In fact, we recall that our procedure is sound but not complete. It is not possible to find a complete method since, in general, model checking ATL in the context of imperfect information and perfect recall strategies is proved to be undecidable \cite{DimaT11}. In this work we have considered how to remove the imperfect information from the models to find decidability, but it would be interesting to consider the other direction that produces undecidability, i.e. the perfect recall strategies. In particular, we could remove the perfect recall strategies to generate games with imperfect information and imperfect recall strategies, a problem in general decidable, and then use CTL model checking to provide a verification result. 
%
Additionally, we plan to extend our techniques here developed to more expressive languages for strategic reasoning including Strategy Logic \cite{ChatterjeeHenzingerPitterman07,DBLP:journals/tocl/MogaveroMPV14}.


\bibliographystyle{unsrt}
\bibliography{new_bib}

\end{document}